\documentclass[11pt]{article}

\usepackage{amsmath,amsbsy,amsfonts,amssymb,amsthm,dsfont,fullpage,units}
\usepackage{graphicx}
\usepackage{authblk}
\usepackage{algorithm,algorithmic,mathtools}
\usepackage{color,cases}
\usepackage{tikz}
\usetikzlibrary{calc,shapes}
\usepackage{subfigure}
\usepackage{hyperref}
\usepackage{nicefrac}

\newcommand{\sfrac}[2]{#1\nicefrac{}{#2}}

\newtheorem{claim}{Claim}
\newtheorem{corollary}{Corollary}[section]

\newtheorem{lemma}{Lemma}

\newtheorem{theorem}{Theorem}[section]
\theoremstyle{definition}

\newtheorem{definition}{Definition}

 \newcommand{\IGNORE}[1]{}

\newcommand{\citep}{\cite}
\newcommand{\citet}{\cite}

\newcommand\E{\mathbb{E}}
\newcommand\R{\mathbb{R}}

\DeclareMathOperator{\diag}{diag}

\newcommand\inner[1]{\langle #1 \rangle}

\newcommand\poly{\operatorname{poly}}
\newcommand\polylog{\operatorname{polylog}}

\newcommand{\Exp}{\mathop{\mathbb E}\displaylimits}

\def\shownotes{1}  %set 1 to show author notes
\ifnum\shownotes=1
\newcommand{\authnote}[2]{{ $\ll$\textsf{\footnotesize #1 notes: #2}$\gg$}}
\else
\newcommand{\authnote}[2]{}
\fi

\newcommand{\pE}{\widetilde{\Exp}}

\newcommand{\sosle}{\preceq}
\newcommand{\sosge}{\succeq}
\title{Decomposing Overcomplete 3rd Order Tensors using Sum-of-Squares Algorithms}
\author[1]{Rong Ge}
\affil[1]{Microsoft Research, New England. Email: rongge@microsoft.com}
\author[2]{Tengyu Ma}
\affil[2]{Princeton University. Email: tengyu@cs.princeton.edu}
\date{}
\begin{document}
\maketitle

\begin{abstract}
Tensor rank and low-rank tensor decompositions have many applications in learning and complexity theory. Most known algorithms use unfoldings of tensors and can only handle rank up to $n^{\lfloor p/2 \rfloor}$ for a $p$-th order tensor in $\R^{n^p}$. Previously no efficient algorithm can decompose 3rd order tensors when the rank is super-linear in the dimension. Using ideas from sum-of-squares hierarchy, we give the first quasi-polynomial time algorithm that can decompose a random 3rd order tensor decomposition when the rank is as large as $n^{3/2}/\poly\log n$.

We also give a polynomial time algorithm for certifying the injective norm of random low rank tensors. Our tensor decomposition algorithm exploits the relationship between injective norm and the tensor components. The proof relies on interesting tools for decoupling random variables to prove better matrix concentration bounds, which can be useful in other settings.
\end{abstract}

\section{Introduction}

Tensors, as natural generalization of matrices, are often used to represent multi-linear relationships or data that involves higher order correlation. A $p$-th order tensor $T\in \mathbb{R}^{n^p}$ is a $p$-dimensional array indexed by $[n]^p$. A tensor $T$ is rank-1 if it can be written as the outer-product of $p$ vectors $T = a_1\otimes \dots \otimes a_p$, where $a_i\in \mathbb{R}^n$ (for $i = 1,\dots,p$). Equivalently, $T_{i_1,\dots,i_p} = \prod_{j=1}^p a_{j}(i_j)$ where $a_j(i_j)$ denotes the $i_j$-th entry of vector $a_j$. 

%For vectors $a_1,...,a_p\in \R^n$, we use tensor product $T = a_1\otimes a_2\otimes \cdots \otimes a_p$ to denote a $p$-th order tensor in the space $R^{n^p}$, such that $T_{i_1,i_2,...,i_p} = \prod_{j=1}^p a_j(i_j)$. The tensor $T = a_1\otimes a_2\otimes \cdots \otimes a_p$ is called a rank-1 tensor. In general, a $p$-th order tensor is just a $p$-dimensional array that may not be expressable as tensor products of vectors. 

Low rank tensors --- similar to low rank matrices --- are widely used in many applications. The rank of tensor $T$ is defined as the minimum number $m$ such that $T$ can be written as the sum of $m$ rank-1 tensors. %, and such a representation is known as CP/PARAFAC decomposition~\cite{carroll1970analysis,harshman1970foundations}. %$T_i$'s, that is, `$T = \sum_{i=1}^m T_i$. 
%Given a tensor $T \in \R^{n^p}$, we would like to decompose $T$ into sum of rank-1 tensors
%In many applications, it is important to consider the {\em low rank} decomposition for tensors. 
%\begin{equation}
%\label{eq:lowrank}
%T = \sum_{i=1}^m a_{i,1}\otimes a_{i,2} \otimes \cdots \otimes a_{i,p}, a_{i,j} \in \R^n.
%\end{equation}
%Representing a tensor $T$ in this form is known as the CP/PARAFAC\cite{carroll1970analysis,harshman1970foundations} decomposition.  The smallest number $m$ where such a decomposition exists is called the rank of the tensor $T$. 
This agrees with the definition of matrix rank. However, most of the corresponding tensor problems are much harder: for $p \ge 3$ computing the rank of the tensor (as well as many related problems) is NP-hard~\cite{haastad1990tensor,TensorNPHard}. Tensor rank is also not as well-behaved as matrix rank (see for example the survey~\cite{comontensorsurvey}).

%Tensors and the notion of tensor rank have found many applications in theoretical computer science including matrix multiplication and circuit complexity\cite{strassen1973vermeidung, brgisser2010algebraic, raz2013tensor}. 
Unlike matrices, low rank tensor decompositions are often unique \cite{Kruskal:77}, which is important in many applications. In special cases (especially when rank $m$ is less than dimension $n$) tensor decomposition can be efficiently computed. Such specialized tensor decompositions have been the key algorithmic ideas in many recent algorithms for learning latent variable models, 
%In this paper we are mostly interested in the applications to learning. 
%Recently tensor decomposition has been applied to learning many latent variable models, 
including mixture of Gaussians, Independent Component Analysis, Hidden Markov Model and Latent Dirichlet Allocation (see \cite{AnandkumarEtal:tensor12}). In many cases tensor decomposition can be viewed as reinterpreting previous spectral learning results \cite{Chang96,MR06,AnandkumarEtal:lda12,AnandkumarHsuKakade:COLT12}. This new interpretation has also inspired many new works (e.g. \cite{AnandkumarEtal:community12COLT,bhaskara2014smoothed,GHK15}).

A common limitation in early tensor decomposition algorithms is that they only work for the undercomplete case when rank $m$ is at most the dimension $n$. 
%in learning latent variable models using tensor decompositions is that we can only compute the decomposition efficiently when $m \le n$ (which we call the undercomplete case). 
Although there are some attempts to decompose tensors in the overcomplete case ($m>n$) \cite{de2007fourth,bhaskara2014smoothed,anderson2013more}, all these works require 4-th or higher order tensors. In many machine learning applications, the number of samples required to accurately estimate a 4-th order tensor is too large. In practice algorithms based on 3rd order tensor are much more preferable. Therefore we are interested in the key question: are there any efficient algorithms for overcomplete 3rd order tensor decomposition?

%It is natural to ask whether it is possible to handle overcomplete ($m>n$) case for $3$rd order tensor, especially since the running time and sample complexity of tensor-based algorithms typically scale exponentially in the order of the tensor.

In the worst case setting, overcomplete 3rd order tensors are not well-understood. 
%understanding the overcomplete setting for third order tensors is a hard problem and only few results are known. 
Kruskal~\cite{Kruskal:77} showed the tensor decomposition is unique when the rank $m \le 1.5n-1$ and the components are in general position, but there is no efficient algorithm known for finding this decomposition. Constructing an explicit 3rd order tensor with rank $\Omega(n^{1+\epsilon})$ will give nontrivial circuit complexity lowerbounds \cite{strassen1973vermeidung}, while the best known rank bound for an explicit 3rd order matrix is only $3n - O(\log n)$~\cite{alexeev2011tensor}.

For many of the learning applications, it is natural to consider the average case problem where the components of the tensor are chosen according to a random distribution. In this case \cite{AltTensorDecomp2014} give a polynomial time algorithm that can find the true components when $m = Cn$ for any constant $C>0$ (however the runtime depends exponentially on $C$). 

%The main difficulty in handling overcomplete 3rd order tensors is that there is no natural {\em unfolding} that can certify the rank of the tensor. An unfolding of a tensor is a mapping from the tensor to a matrix. We can unfold a 4-th order tensor $T$ into a matrix $M$ of size $n^2\times n^2$ where $M_{(i_1,i_2), (i_3,i_4)} = T_{i_1,i_2,i_3,i_4}$. \Tnote{Not clear what the following sentence means}Such an unfolding can have rank up to $n^2$ (it is easy to see the rank of the tensor is at least the rank of an unfolding). In general from the unfolding we can get lowerbounds for tensor ranks up to $n^{\lfloor p/2\rfloor}$, so $p = 3$ is the only case where this method does not give any nontrivial bound.

This paper also considers this average case setting and gives a quasi-polynomial algorithm for decomposing the tensor when $m$ can be as large as $n^{3/2}$. The main idea of the algorithm is based on sum-of-squares (SoS) SDP hierarchy (\cite{parrilo2000structured, lasserre2001global}, see Section~\ref{sec:prelim} and the recent survey \cite{BarakS14}). The main difficulty in handling overcomplete 3rd order tensors is that there is no natural {\em unfolding} (i.e. mapping to a matrix) that can certify the rank of the tensor. We can unfold a 4-th order tensor $T$ into a matrix $M$ of size $n^2\times n^2$ where $M_{(i_1,i_2), (i_3,i_4)} = T_{i_1,i_2,i_3,i_4}$. However, unfolding 3rd order tensor will result in a very unbalanced matrix of dimension $n\times n^2$ that cannot have rank more than $n$. Intuitively, the power of SoS-based algorithm is that it can provide higher-order ``pseudo-moments'' that will allow us to use nontrivial unfoldings.% will behave similarly to higher-order tensor, even though we only have access to 3rd order tensor. 

%\Tnote{Not clear what the following sentence means}Such an unfolding can have rank up to $n^2$ (it is easy to see the rank of the tensor is at least the rank of an unfolding). In general from the unfolding we can get lowerbounds for tensor ranks up to $n^{\lfloor p/2\rfloor}$, so $p = 3$ is the only case where this method does not give any nontrivial bound.

In particular, the key component of the proof is a way of certifying {\em injective norm} (see Section~\ref{sec:prelim}) of random tensors, which is closely related to the problem of certifying the 2-to-4 norm of random matrices\cite{BBH12}. Recently, there has been an increasing number of applications of SoS hierarchy to learning problems. \cite{BarakKS14} give algorithms for finding the sparsest vectors in a subspace, which is closely related to many learning problems. \cite{BKS14} give a new algorithm for dictionary learning that can handle nearly linear sparsity, and also an algorithm for robust tensor decomposition.However their result requires a tensor of high order.
\cite{BarakM15} studies a related problem of tensor prediction, also using ideas of SoS hierarchies.

\subsection{Our Results}

In this paper we give a quasi-polynomial time algorithm for decomposing third-order tensors when the rank $m$ is almost as large as $n^{3/2}$ and the components of the tensor is chosen randomly. More concretely, we define $\mathcal{D}_{m,n}$ to be a distribution of third order tensors of the following form:

$$
T = \sum_{i=1}^m a_i^{\otimes 3},
$$
where the vectors $a_i\in \R^n$ are uniformly random vectors in $\{\pm \frac{1}{\sqrt{n}}\}^n$ and $a_i^{\otimes 3}$ is short for $a_i\otimes a_i\otimes a_i$. Our goal is to recover these components $a_i$'s. Since any permutation of $a_i$'s is still a valid solution, we say two decompositions are $\epsilon$-close if they are close after an arbitrary permutation:

\begin{definition}[$\epsilon$-close] Two sets of vectors $\{a_i\}_{i\in[m]}$ and $\{\hat{a}_i\}_{i\in [m]}$ in $\R^n$ are $\epsilon$-close if there exists a permutation $\pi:[m]\to [m]$ such that $\|\hat{a}_{\pi(i)} -  a_i\| \le \epsilon$. Two decompositions of the tensor $T$ are $\epsilon$-close if their components are $\epsilon$-close.
\end{definition}

For tensors in distribution $\mathcal{D}_{m,n}$ our algorithm can recover the decomposition as long as $m\ll n^{3/2}$.

\begin{theorem}
\label{thm:main}
Given a tensor $T =  \sum_{i=1}^m a_i^{\otimes 3}$ sampled from distribution $\mathcal{D}_{m,n}$, when $m \ll n^{3/2}$  there is an algorithm that runs in time $n^{O(\log n)}$ and with high probability returns a decomposition $T \approx  \sum_{i=1}^m \hat{a}_i^{\otimes 3}$ that is $0.1$-close to the true decomposition. \end{theorem}

Our result easily generalizes to many other distributions for $a_i$ (including a uniform random vector in unit sphere or a spherical Gaussian).

The algorithm does not output a very accurate solution (the accuracy can be improved to $\epsilon$ with an exponential dependency on $1/\epsilon$). However it is known that alternating minimization algorithms can refine the decomposition once we have a nice initial point\cite{AltTensorDecomp2014}:

\begin{theorem}[\cite{AltTensorDecomp2014}]\label{thm:refine} Given a tensor $T$ from distribution $\mathcal{D}_{m,n}$ ($m\ll n^{3/2}$), and an initial solution that is $0.1$-close to the true decomposition, then for any $\epsilon > 0$ (that may depend on $n$) there is an algorithm that runs in time $\poly(n,\log 1/\epsilon)$ that with high probability finds a refined decomposition that is $\epsilon$-close to the true decomposition.
\end{theorem}

Combining the two results we have an algorithm that runs in time $n^{O(\log n)}\poly\log (1/\epsilon)$ that recovers a decomposition that is component-wise $\epsilon$-close to the true decomposition.

\begin{corollary}
\label{cor:refine}
Given a tensor $T =  \sum_{i=1}^m a_i^{\otimes 3}$ sampled from distribution $\mathcal{D}_{m,n}$, when $m \ll n^{3/2}$  for any $\epsilon > 0$ there is an algorithm that runs in time $n^{O((\log n))}\poly\log(1/\epsilon)$ and with high probability returns a decomposition $T \approx  \sum_{i=1}^m \hat{a}_i^{\otimes 3}$ that is $\epsilon$-close to the true decomposition.
\end{corollary}

The main idea in proving Theorem~\ref{thm:main} is the observation that when the tensor is generated randomly from $\mathcal{D}_{m,n}$, the true components are close to the maximizers of the multilinear form $T(x,x,x) = \sum_{i,j,k\in[n]} T_{i,j,k} x_ix_jx_k = \sum_{i=1}^m \inner{a_i,x}^3$. The maximum value of $T(x,x,x)$ on unit vectors $\|x\| = 1$ is known as the injective norm of the tensor. Computing or even approximating the injective norm is known to be hard \cite{Gurvits,harrow2013testing}. A key component of our approach is a sum-of-square algorithm (see Section~\ref{sec:prelim} for preliminaries about sum-of-square algorithms) that certifies that the injective norm of a random tensor from $\mathcal{D}_{m,n}$ is small.

\begin{theorem}
\label{thm:maininj}
For a tensor $T$ in distribution $\mathcal{D}_{m,n}$, when $m \ll n^{3/2}$ with high probability the injective norm of $T$ is bounded by $1+o(1)$. Further, this can be certified in polynomial time.
\end{theorem}

Our results (Theorem~\ref{thm:main} and \ref{thm:maininj}) still hold when we are given a tensor $\tilde{T}$ that is $1/\poly(n)$-close to $T$ in the sense that the spectral norm of an unfolding of $\tilde{T} - T$ is $O(1/\poly\log(n))$. Theorem~\ref{thm:refine} (and hence Corollary~\ref{cor:refine}) requires a tensor $\tilde{T}$ such that the unfolding of $\tilde{T} - T$ has spectral norm bounded by $\epsilon/\poly(n)$.

\paragraph{Organization} The rest of this paper is organized as follows: In Section~\ref{sec:prelim} we introduce tensor notations and SoS hierarchies. Then we describe the main idea of the proof which relates tensor decomposition to the injective norm of tensor (Section~\ref{sec:mainidea}). In Section~\ref{sec:certify} we give a polynomial time algorithm for certifying the injective norm of a random 3rd order tensor. Using this as a key tool in Section~\ref{sec:alg} we present the quasi-polynomial time algorithm that can decompose randomly generated tensors when $m \ll n^{3/2}$. 

\section{Preliminaries}
\label{sec:prelim}
\paragraph{Notations} In this paper we use $\|\cdot \|$ to denote the $\ell_2$ norm of vectors and the spectral norm of matrices. That is, $\|v\| = \sqrt{\sum_i v_i^2}$ and $\|A\| = \sup_{\|u\| = 1} \|Au\|$. Note that we will be using the sum-norm instead of expectation norm $\|v\|_{exp} = \sqrt{\E_i[v_i^2]}$ because the scaling of sum-norm is more natural for the tensor decomposition setting. We use $\inner{u,v}$ to denote the inner product of $u$ and $v$. When $A$ and $B$ are two matrices, we use standard notation $A\preceq B$ to denote the fact that $B-A$ is a positive semidefinite. For a $m\times n$ matrix $U$ and a $p\times q$ matrix $V$, we define the Kronecker product $U\otimes V$ as the $mp\times nq$ block matrix 
$$U\otimes V = \left[\begin{matrix}
U_{1,1}V &  \cdots & U_{1,n}V \\
\vdots  & \ddots & \vdots  \\
U_{m,1}V &  \cdots & U_{m,n}V
\end{matrix}\right]$$

We use $\tilde{O}$ notations to hide dependencies on $\polylog$ factors in $n$ and $m$. When we write $f \ll g$ we mean $f \le g/O(\poly\log n)$. Throughout the paper high probability means the probability is at least $1- n^{-\omega(1)}$.

\paragraph{Tensors} Tensors are multi-dimensional arrays. In this paper for simplicity we only consider 3rd order symmetric tensors and their symmetric decompositions. For a third order symmetric tensor $T$, the value of $T_{i,j,k}$ only depends on the multi-set $\{i,j,k\}$, so $T_{i,j,k} = T_{j,i,k} = T_{k,i,j}$ (and more generally all the 6 permutations are equal). For a vector $v\in \R^n$, we use $v^{\otimes 3} \in \R^{n^3}$ to denote the symmetric third order tensor such that $v^{\otimes 3}_{i,j,k} = v_iv_jv_k$. Our goal is to decompose a tensor $T$ as $T = \sum_{i=1}^m a_i^{\otimes 3}$. 

There is a bijection between 3rd order symmetric tensors and homogeneous degree 3  polynomials. In particular, for a tensor $T$ we define its corresponding polynomial $T(x,x,x) = \sum_{i,j,k=1}^n T_{i,j,k} x_ix_jx_k$. It is easy to verify that if $T = \sum_{i=1}^m a_i^{\otimes3}$ then $T(x,x,x) = \sum_{i=1}^m \inner{a_i,x}^3$.

The injective norm $\|T\|_{inj}$ is defined to be the maximum value of the corresponding polynomial on the unit sphere, that is:
$$
\|T\|_{inj} := \sup_{\|x\| = 1} T(x,x,x).
$$

It is not hard to prove when $m \ll n^{3/2}$, and the tensor $T$ is chosen from the distribution $\mathcal{D}_{m,n}$, with high probability $1-o(1) \le \|T\|_{inj} \le 1+o(1)$, and in fact the value $T(x,x,x)$ is only close to $1$ if $x$ is close to one of the components $a_i$. We will give a (SoS) proof of this fact in Section~\ref{sec:alg}

\paragraph{Sum-of-Square Algorithms and Proofs} Here we will only briefly introduce the notations and key concepts that are used in this paper, for more detailed discussions and references about SoS proofs we refer readers to \cite{BarakS14} (especially Section 2). 

Sum-of-squares proof system is a proof system for polynomial equalities and inequalities. Given a set of constraints $\{r_i(x) = 0\}$, and a degree bound $d$, we say there is a degree $d$ SoS proof for $p(x) \ge q(x)$ if $p(x) - q(x)$ can be written as a sum of squares of polynomials modulo $r_i(x) = 0$, as defined formally below. 

\begin{definition} [SoS proof of degree $d$]
	For a set of constraints $R = \{r_1(x) = 0,\dots,r_{t}(x) = 0\}$, and an integer $d$, we write $$p(x)\sosge_{R,d}q(x)$$ if there exists polynomials $h_i(x)$ for $i = 0,1,\dots, \ell$ and $g_j(x)$ for $j = 1,\dots,t$ such that $\deg(h_0^2(p(x)-q(x))) \le d$, $\deg(h_i)\le d/2$ (for $i > 0$) and $\deg(g_jr_j)\le d$ that satisfy
	$$h_0(x)^2(p(x)-q(x)) = \sum_{i=1}^{\ell}h_i(x)^2 + \sum_{j=1}^{t}r_j(x)g_j(x),$$
	We will drop the subscript $d$ when it is clear form the context. % and when $R =\{r(x) = 0\}$, we will simply write $\sosge_r$ instead of $\sosge_R$. 
\end{definition}

Note that the constraints set can be easily generalized to a set of inequalities by adding auxiliary variables. For example, constraint $r(x)\ge 0$ can be implemented as $r(x) = z^2$ where $z$ is an auxiliary variable. 

%We say degree $d$ SoS {\em refutes} a set of constraints $R$ if $-1\sosge_{R,d} 0$. It is well-known\cite{parrilo2000structured,lasserre2001global} that a refutation of degree $d$ can be found in $\poly(tn^d)$ time (where $n$ is the number of variables). 
Many well-known inequalities can be proved using a low degree SoS proof, among them the most useful and important one is Cauchy-Schwarz inequality, which can be proved via degree-2 sum of squares.  Another one is that $x^TAx \sosle \|A\|\|x\|^2$. This is pretty useful when $A$ is a random matrix where we can use random matrix theory to bound the spectral norm of $A$. %We state several others of them in Appendix~\ref{} that will be useful in later proofs.

In order to turn an SoS arguments into an algorithm, we often consider the {\em pseudo-expectation}. Just as we have expectations for real distributions, we think of pseudo-expectation as expectations for pseudo-distributions that cannot be distinguished from true expectations using low degree polynomials. Pseudo-expectation can be viewed as a dual of SoS refutations.

\begin{definition} [pseudo-expectation]
	A degree $d$ pseudo-expectation $\pE$ is a linear operator that maps degree $d$ polynomials to reals. The operator satisfies $\pE[1] = 1$ and $\pE[p^2(x)] \ge 0$ for all polynomials $p(x)$ of degree at most $d/2$.
	We say a degree-$d$ pseudo-expectation $\pE$ satisfies a set of equations $\{r_i(x):i=1\dots,\ell\}$ if for any $i$ and any $q(x)$ such that $\deg(r_iq)\le d$, 
	$$\pE\left[r_i(x)q(x)\right] = 0$$
	%\Tnote{Need a version of inequality constraints.. }
\end{definition}

By definition, if $p(x)\sosle_{R,d} q(x)$, and degree-$d$ pseudo-expectation satisfies $R$, then we can take pseudo-expectation on both sides and obtain $\pE\left[p(x)\right]\le \pE\left[q(x)\right]$. We will use this property of pseudo-expectation many times in the proofs. 

The relationship between pseudo-expectations and SoS refutations can be summarized in the following informal lemma:

\begin{lemma}[~\cite{parrilo2000structured,lasserre2001global}, c.f.~\cite{BarakS14}, informal stated] For a set of constraints $R$, either there is an SoS refutation of degree $d$ that refutes $R$, or there is a degree $d$ pseudo-expectation that satisfies $R$. Such a refutation/pseudo-expectation can be found in $\poly(tn^d)$ time.
\label{lem:sos}
\end{lemma}

\section{Relating Tensor Decompositions and Injective Norm}
\label{sec:mainidea}
In this section we introduce the main idea of our proof. Given a tensor $T = \sum_{i=1}^m a_i^{\otimes 3}$ from distribution $\mathcal{D}_{m,n}$, we first make some observations about its corresponding polynomial $T(x,x,x) = \sum_{i=1}^m \inner{a_i,x}^3$.

When $x = a_1$, we know $T(a_1,a_1,a_1) = 1 + \sum_{i=2}^m \inner{a_i,a_1}^3$. Here conditioned on $a_1$, the second term is a sum of independent random variables ($\inner{a_i,a_1}^3$). By the distribution $\mathcal{D}_{m,n}$ we know these variables have mean 0 and absolute value around $1/n^{3/2}$. Standard concentration bounds show when $m \ll n^{3/2}$ with high probability $T(a_1,a_1,a_1) = 1\pm o(1)$. 

On the other hand, suppose $x$ is a random vector in the unit sphere, then $T(x,x,x) = \sum_{i=1}^m \inner{a_i,x}^3$ is again a sum of random variables. By concentration bounds we know for any particular $x$, when $m\ll n^{3/2}$ with high probability $T(x,x,x) = o(1)$. This can actually be generalized to all vectors $x$ that do not have large correlation with $a_i$'s using $\epsilon$-net arguments. %Therefore we observe

\newtheorem*{observation}{Observation}

\begin{observation} For a random tensor $T\sim \mathcal{D}_{m,n}$, when $m = n^{3/2}$ with high probability $T(x,x,x) \le 1+o(1)$ for $\|x\| = 1$. Further when $T(x,x,x)$ is close to $1$ the vector $x$ is close to one of the components $a_i$'s. 
\end{observation}

Later we will give a SoS proof for this observation. Based on this observation, if we want to find a component, then it suffices to find a vector $x$ such that $T(x,x,x)$ is close to $1$. Using the idea of pseudo-expectations, we can do this in two steps:

\begin{enumerate}
\item Find a pseudo-expectation $\pE[x]$ that satisfies the constraint $\|x\|^2 - 1 = 0$ and maximizes $\pE[T(x,x,x)]$.
\item ``Sample'' from this pseudo-distribution with psuedo-expectations $\pE$ to get a vector $x$ such that $T(x,x,x) \approx 1$, in particular $x$ will be close to one of the components $a_i$'s.
\end{enumerate}

In Section~\ref{sec:certify} we will prove the first part of the observation. In particular we show even though we are maximizing over pseudo-expectation $\pE[x]$ (instead of real distributions over $x$), we can still guarantee the maximum value $\pE[T(x,x,x)]$ is at most $1+1/\log n$ with high probability.

In Section~\ref{sec:alg} we give algorithms for finding a component given a pseudo-expectation $\pE$ with $\pE[T(x,x,x)] \approx 1$. The main idea of our algorithm is similar to the robust tensor decomposition algorithm in \cite{BKS14}: first we show there must be a component $a_i$ such that $\pE[\inner{a_i,x}^d]$ is large for a large $d$, then we use ideas in \cite{BKS14} to find the component $a_i$.

\section{Certifying Injective Norm}
\label{sec:certify}
\begin{algorithm}\caption{Certifying Injective Norm}\label{alg:certifying-norm}
	\begin{algorithmic}[]
		\REQUIRE A random 3-tensor $T$
		\ENSURE If $\|T\|_{\textrm{inj}} > 1+1/\log n$, return NO. If $T\sim \mathcal{D}_{m,n} (m\ll n^{3/2})$, then w.h.p. return YES. \\
		
		\STATE Solve the following optimization and obtain optimal value $\textrm{OPT}$
		\begin{eqnarray}
		\textrm{Maximize} && \pE \left[T(x,x,x)\right] \nonumber\\
		\textrm{Subject to } && \pE \textrm{ is a degree-12 pseudo-expectation}\\
		&& \textrm{that satisfies $\{r(x) = \|x\|^2 - 1 = 0\}$}
		%\pE(r(x)g(x)) = 0 \textrm{ for } r(x) = \|x\|^2 - 1 \textrm{ and any $g$ with $\deg(g)\le 10$}
		%\textrm{that is consistent with $r(x) = \|x\|^2  - 1 = 0$}
		\end{eqnarray}
		
		\RETURN YES if $ \textrm{OPT} \le 1+1/\log n$ and NO otherwise. 
		\end{algorithmic}
	\end{algorithm}
	%pseudo 
%%\textbf{Initialize} $\tilA^
%$ that is $(\delta_0, 2)$-near to $\trA$
%
%\textbf{Repeat}  for $s = 0, 1, ..., T$
%\vspace{-0.1in}
%\begin{align*}
%\textbf{\textbf{Decode: }} & \tilx = \textrm{threshold}_{C/2}((\tilA^s)^Ty)  \textrm{ for each sample $y$}\nonumber\\
%%\textrm{ for $i = 1, 2, ... , p$}\nonumber\\
%\textbf{\textbf{Update: }} & \tilA^{s+1} = \tilA^s - \eta g^{s} \mbox{ where } g^{s} = \Exp[(y - \tilA^s \tilx) \tilx^T] \quad \nonumber\\
%%= \frac{1}{p}\sum_{i=1}^p (y^{(i)} - \tilA^s \tilx^{(i)}) (\tilx^{(i)})^T \quad \nonumber\\
%\textbf{\textbf{Project: }}& \tilA^{s+1} = \mbox{Proj}_{\mathcal{B}} \tilA^{s+1} \mbox{(where $\mathcal{B}$ is defined in Definition~\ref{def:whiten})}\nonumber
%\end{align*}

%\textbf{Input: } A random 3-tensor $T$\\
%
%\textbf{Output: } If $|T|_{\textrm{inj}} > 1+\epsilon$, returns 0. If $T\sim \mathcal{D}_{m,n}$, then with high probability returns 1. \\
%\vspace{0.1in}
%\noindent Solve the convex optimization problem: 

In this section, we give Algorithm~\ref{alg:certifying-norm} based on SoS hierarchy that certifies the injective norm of random tensor. In particular, we will prove Theorem~\ref{thm:maininj} which we restate in more details here.

\begin{theorem}\label{thm:maininjdet}
Algorithm~\ref{alg:certifying-norm}  always returns NO when $\|T\|_{\textrm{inj}} > 1+1/\log n$. When $T\sim \mathcal{D}_{m,n}$ and $m\ll n^{3/2}$, Algorithm~\ref{alg:certifying-norm} returns YES with high probability over the randomness of $T$. Further, the same guarantee holds given an approximation $\tilde{T}$ where if $M \in \R^{n\times n^2}$ is an unfolding of $T - \tilde{T}$, $\|M\| \le 1/2\log n$. 
\end{theorem}

When $\|T\|_{\textrm{inj}} > 1+1/\log n$, then by definition there must be a vector $x^*$ that satisfies $\|x^*\| = 1$ and $T(x^*,x^*,x^*) > 1/\log n$. We can take $\pE$ to be the expectation of a distribution that is only supported on $x^*$ (i.e. with probability $1$ $x = x^*$). Clearly this pseudo-expectation is valid, and OPT will be at least larger than $1/\log n$. Hence the algorithm returns NO.

For random tensor $T$, we hope to show that with high probability, the tensor norm is less than $1+1/\log n$ can be proved via SoS. 
\begin{theorem}\label{thm:proof}
With high probability over the randomness of the tensor $T$, for $r(x) = \|x\|^2 - 1$, % for $\delta = $, %there exists polynomial $g(x)$, and $h_1(x),\dots,h_{\ell}(x)$  with degree at most $??$ such that 
\begin{equation}
T(x,x,x) \sosle_{r,12} 1+\widetilde{O}(m/n^{3/2}) \label{eqn:sos-proof-main}
\end{equation}
\end{theorem}

Note that taking pseudo-expectation $\pE$ on both hand sides of~\eqref{eqn:sos-proof-main}, for any degree-12 pseudo-expectation $\pE$ that is consistent with $r(x)$, 

$$\pE\left[T(x,x,x)\right] \le 1 + \widetilde{O}(m/n^{3/2})$$

That is, when $m\ll n^{3/2}$, the objective value of the convex program in Algorithm~\ref{alg:certifying-norm} is less than $1+1/\log n$ with high probability for random tensor. 

Now we need to prove Theorem~\ref{thm:proof}. We first use Cauchy-Schwarz inequality to transform LHS of~\eqref{eqn:sos-proof-main} to a degree-4 polynomial, which would then correspond to 4th order tensors and enable non-trivial unfoldings. % is a degree-3 polynomial. Typically even degree polynomial is easier for us to handle because we could potentially view it as a degree two polynomial and apply random matrix theory. %To %this end, we note that if $p(x)^2\sosle_{R,k} a^2$ for positive real number $a$, then by simple manipulations we have $$p(x) - a \sosle_{R,k}\frac{1}{2a}(p(x)-a)^2\sosle_{R,k'}0$$
%where $k' = \max\{k,2\deg(p)\}$. 
%By Lemma~\ref{lem:square-root}, it suffices to bound the square of the LHS of~\eqref{eqn:sos-proof-main} by $1+o(1)$ using SoS proofs. All of our proofs in this section only involve constant degrees. So we drop the subscript for degree in this section. 

\begin{claim}\label{clm:cs}
\begin{equation}
[T(x,x,x)]^2 \sosle_{r,12} \underbrace{\sum_{i=1}^m \inner{a_i,x}^4}_{\textrm{2-4 norm}} + \underbrace{\sum_{i\neq j}\inner{a_i,a_j}\inner{a_i,x}^2\inner{a_j,x}^2}_{:= p(x)}. \label{eqn:two-terms}
\end{equation}
\end{claim}

\begin{proof}
This is a direct application of Cauchy-Schwarz inequality:
%We start by writing Cauchy–Schwarz and %and reduce the degree of the polynomial that we want to bound to 4: 
%
\begin{align*}
\left(T\cdot x^{\otimes 3}\right)^2 &=  \left(\sum_{i=1}^m \inner{a_i,x}^3\right)^2 = \left\langle\sum_{i=1}^m\inner{a_i,x}^2a_i, x\right\rangle^2 \sosle \left\|\sum_{i=1}^m\inner{a_i,x}^2a_i\right\|^2\|x\|^2 \sosle_r \left\|\sum_{i=1}^m\inner{a_i,x}^2a_i\right\|^2
\end{align*}

Expanding this quantity, and using the fact that $\|a_i\| = 1$, we get 

\begin{align}
\left\|\sum_{i=1}^m\inner{a_i,x}^2a_i\right\|^2 & = \sum_{i=1}^m \inner{a_i,x}^4 + \sum_{i\neq j}\inner{a_i,a_j}\inner{a_i,x}^2\inner{a_j,x}^2.
\end{align}
\end{proof}

The first term is closely related to $2$-to-$4$ norm of random matrices: let $A\in \R^{m\times n}$ be a matrix whose rows are equal to $a_i$'s, then $\|A\|_{2\to 4} = \sup_{\|x\| = 1} \|Ax\|_4$. Clearly, $\|A\|_{2\to 4}^4 = \sup_{\|x\| = 1} \sum_{i=1}^m \inner{a_i,x}^4$ is the maximum value of the first term. This is considered in \cite{BBH12} where they gave a SoS proof that when $m\ll n^2$ the first term is bounded by $O(1)$. Here we are in the regime $m \ll n^{3/2}$ so we can improve the bound to $1+o(1)$ (The proof is deferred to Appendix~\ref{sec:append:2to4}):

%We bound the two terms in RHS of~\eqref{eqn:two-terms} separately by the following two lemmas. Note that the first term on RHS is just 2$\rightarrow$4 norm of the random matrix with $a_i$ as rows. \cite{BBH12} showed that when $m \ll n^2$, 
%\begin{equation}
%\sum_{i=1}^m \inner{a_i,x}^4\le O(1) \label{eqn:weak-2-4-norm}
%\end{equation}
%
%We improve this bound to $1+o(1)$ when $m\ll n^{3/2}$. The proof of the following lemma is given in Section~\ref{sec:2-4norm}. 
%
\begin{lemma}\label{lem:2-4norm}With high probability over the randomness of $a_i$'s, 
	\begin{equation}
	\sum_{i=1}^m \inner{a_i,x}^4\sosle_{r,12} 1+\widetilde{O}(m/n^{3/2}) \label{eqn:2-4norm}\end{equation}
\end{lemma}

The harder part of the proof is to deal with the second term $p(x)$ on the RHS of~\eqref{eqn:two-terms}. The naive idea would be to let $y= x^{\otimes 2}$ and view $p(x)$ as a degree-2 polynomial of $y$, 
\begin{equation}
q(y)=\sum_{i\neq j}\inner{a_i,a_j}\inner{a_i\otimes a_i,y}\inner{a_j\otimes a_j,y} = y^TNy. \label{eqn:view}
\end{equation}
Here $N$ is an $n^2$ by $n^2$ random matrix that depends on $a_i$'s. Suppose $N$ has spectral norm less than $o(1)$, then we have $y^TNy \sosle \|N\|\|y\|^2$, and by replacing $y = x\otimes x$ we obtain $p(x) = q(x\otimes x)\sosle o(1)$. However, in our case the matrix $N$ have spectral norm much larger than $o(1)$. %Note that the problem is that we are trying to prove a bound $q(y)\sosle o(1)$ while actually we only wanted to prove a weaker result $q(x\otimes x)\sosle o(1)$, and the former one turns out to be not true. 

Our key insight is that we could have different ways to unfold $p(x)$ into a degree-2 polynomial. In particular, we use the following way of unfolding: 

\begin{equation}
q'(y) = \sum_{i\neq j}\inner{a_i,a_j}\inner{a_i\otimes a_j,y}\inner{a_i\otimes a_j,y} = y^TMy \label{eqn:better-view}
\end{equation}

where $M$ is the $n^2$ by $n^2$ matrix that encodes the coefficients of $q'(y)$, 
$$M        = \sum_{i\neq j}\inner{a_i,a_j} (a_i\otimes a_j) (a_i\otimes a_j)^T $$

It turns out that $q'(y)$ still have the property that $q'(x\otimes x) = p(x)$. The matrix $M$ has much better spectral norm bound, which leads us to the bound for $p(x)$. 

\begin{lemma}\label{lem:px}
When $m\ll n^{3/2}$, the matrix $M = \sum_{i\neq j}\inner{a_i,a_j} (a_i\otimes a_j) (a_i\otimes a_j)^T$ has spectral norm at most $\widetilde{O}(m/n^{3/2})$ and as a direct consequence, 
$$p(x)\sosle_{r,4} \widetilde{O}(m/n^{3/2})$$
\end{lemma}

First we give an informal and suboptimal bound for intuition. Let $B$ be the $n^2\times m^2$ matrix whose $(i,j)$-column $(i,j\in [m])$ is $a_i\otimes a_j$ (viewed as an $n^2$ dimensional vector). Then $M$ can be written as $M = B\diag(\inner{a_i,a_j})_{i\neq j}B^T$. Note that $B$ can also be written as $A\otimes A$ where $\otimes$ is the Kronecker product of two matrices, so we have $\|B\|= \|A\|^2 \lesssim m/n$. Then we can bound the norm of $M$ by $\|M\| \le \|B\|\|\diag(b)\|\|B\|\le (m/n)\cdot \max_{i,j}|\inner{a_i,a_j}|\cdot (m/n) \lesssim m^2/n^{5/2}$, where we used the incoherence of $a_i$'s, that is, $|\inner{a_i,a_j}|\lesssim 1/\sqrt{n}$. 
This will only be $o(1)$ when $m \lesssim n^{1.25}$. %which is what we could prove just using incoherence.
%using the form (\ref{eqn:form-M}) of $M$, we can have a simple bound which gives suboptimal result: We have $\|M\| \le \|B\|\|\diag(b)\|\|B\|$. Recalling that $B = A\otimes A$, It is not hard to see that $\|B\| = \|A\|^2\le m/n$. Recall that $b_{ij} = \inner{a_i,a_j}$ is less than $1/\sqrt{n}$ by incoherence, then we have $\|\diag(b)\|\lesssim 1/\sqrt{n}$. Therefore, we have $\|M\|\le m^2/n^{5/2}$. 

Intuitively, this proof is not tight because we ignored potential cancellation caused by the randomness of $ \inner{a_i,a_j}$. Note that $\inner{a_i,a_j}$ have expectation 0, but  we treated them all as positive $1/\sqrt{n}$. If we assume that $\inner{a_i,a_j}$'s are independent $\pm 1/\sqrt{n}$, then $M = \sum_{i\neq j}\inner{a_i,a_j} (a_i\otimes a_j) (a_i\otimes a_j)^T $ would be a sum of PSD matrices with random weights and we can apply more standard matrix concentration bounds to make sure cancellations happen. 

%$$M = \sum_{i\neq j} b_{ij}B_{ij}B_{ij}^T$$

%%where $B_{ij}$ is the $(i,j)$-th column of $B$ (which is $a_i\otimes a_j$  actually). Therefore we can see that $M$ is a sum of random matrices, and $b_{ij}$ introduces some cancellation effect so that the spectral norm bound of $M$ is smaller than 
%Then one should expect that the cancellation effect caused by $\inner{a_i,a_j}$ will make the norm of $M$ smaller than the case when $\inner{a_i,a_j}$ are all of the value $1/\sqrt{n}$. 

However, $\inner{a_i,a_j}$ are of course not independent and our key idea is to decouple the randomness of $\inner{a_i,a_j}$. 

%The main difficulty in proving this matrix concentration bound is that $M$ is the sum of $m^2$ terms that are highly {\em correlated}. In order to remove this correlation we exploit the symmetric properties of the distribution, and use a {\em decoupling} technique developed in \cite{decoupling}.

\begin{proof}(Sketch)
We first replace the vectors $a_i$'s with $\sigma_i a_i$ where $\sigma_i$ is a random $\pm 1$ variable. This is OK because the distribution of $a_i$ and $\sigma_i a_i$ are the same. Now we first sample the $a_i$'s, conditioned on the samples $M = \sum_{i\neq j} \sigma_i\sigma_j\inner{a_i,a_j} (a_i\otimes a_j) (a_i\otimes a_j)^T$ (where only $\sigma_i$'s are still random). Now since the vectors $a_i$'s are all fixed, the correlation between different terms only depends on scalar variables $\sigma_i\sigma_j$, and we never use the term $\sigma_i^2$ (because $i\ne j$).

By a result of \cite{decoupling}, in this case we can decouple the product $\sigma_i\sigma_j$. In particular, in order to prove concentration properties for $M$, it suffices to prove concentration for a different matrix $\sum_{i\neq j} \sigma_i\tau_j\inner{a_i,a_j} (a_i\otimes a_j) (a_i\otimes a_j)^T$. Here $\tau\in\{\pm 1\}^m$ is an independent copy of $\sigma_i$'s. In this way we have decoupled the randomness in $\sigma_i$ and $\tau_i$, and the rest of the Lemma can follow from careful matrix concentration analysis.
\end{proof}

We give the full proof of Lemma~\ref{lem:px} in Appendix~\ref{sec:px}.

\paragraph{Proof Sketch of Main Theorem}

Theorem~\ref{thm:proof} follows directly from Lemma~\ref{lem:2-4norm} and Lemma~\ref{lem:px}. Using Lemma~\ref{lem:sos}, we get the main Theorem~\ref{thm:maininjdet} in the noiseless case. When there is noise, since we have bounds on spectral norm of an unfolding of $\tilde{T}-T$, it implies (by Lemma~\ref{lem:sosspectral}) $[\tilde{T}-T](x,x,x) \sosle_{r,12} 1/2\log n$.it is easy to verify that $\tilde{T}(x,x,x) = T(x,x,x)+[\tilde{T}-T](x,x,x) \sosle_{r,12} 1+ 1/\log n$, so Theorem~\ref{thm:maininjdet} still holds. We give more details in Appendix~\ref{sec:append:certifymain}.
%Note that the first term on RHS is just 2$\rightarrow$4 norm of the random matrix with $a_i$ as rows. ~\cite{DBLP:journals/corr/abs-1205-4484} showed that 
%
%\begin{equation}
%\sum_{i=1}^m \inner{a_i,x}^4 \sosle_r O(1) \label{eqn:BGH-2-4norm}
%\end{equation}

%However, this is not tight enough for our purpose. We pro
%we really want the right hand side to be $1+o(1)$. We are going to have a tighter bound in Section~\ref{sec:2-4norm} and then we bound the second term $p(x)$  by $o(1)$ in Section~\ref{sec:px}

\section{Quasi-polynomial Time Algorithm for Tensor Decomposition}
\label{sec:alg}
In this section we give a quasi-polynomial time algorithm for decomposing random 3rd order tensors in distribution $\mathcal{D}_{m,n}$. In particular, we prove Theorem~\ref{thm:main} which we restate with more details below:

\begin{theorem}\label{thm:quasi-decomp-algo}
Let $T$ be a tensor chosen from $\mathcal{D}_{m,n}$, when $m\ll n^{3/2}$ with high probability over the randomness of $T$  Algorithm~\ref{alg:tensor-decomp}	returns $\{\hat{a}_i\}$ that is $0.1$-close to $\{a_i\}$ in time $n^{O(\log n)}$. Further, the same guarantee holds given an approximation $\tilde{T}$ where if $M \in \R^{n\times n^2}$ is an unfolding of $T - \tilde{T}$, $\|M\| \le 1/10\log n$. 
\end{theorem}

A key component of our algorithm is a way of sampling pseudo-distributions given in \cite{BKS14}:

\begin{theorem}[Theorem 5.1 in \cite{BKS14}]\label{thm:BKS} For every $k\ge 0$, there exists a randomized algorithm with running time $n^{O(k)}$ and success probability $2^{-k/\poly(\epsilon)}$ for the following problem: Given a degree-$k$ pseudo distribution $\{u\}$ over $\R^n$ that satisfies the polynomial constraint $\|u\|^2 = 1$ and the condition $\pE[\inner{c,u}^k] \ge e^{-\epsilon k}$ for some unit vector $c\in \R^n$, output a unit vector $c'\in \R^n$ with $\inner{c,c'} \ge 1-O(\epsilon)$.
\end{theorem}

The basic idea of Algorithm~\ref{alg:tensor-decomp} is as follows. At each iteration, the algorithm tries to find a new vector $\hat{a}_i$. As we discussed in Section~\ref{sec:mainidea}, in order to find a vector close to $a_i$ it finds a vector $x$ with large $T(x,x,x)$ value. Moreover, It enforces that the new vector is different from all previous found vectors by the set of polynomial equations $\{\inner{s,x}^2\le 1/8: s\in S\}$. Intuitively, if we haven't found all of the vectors $a_i$'s any of the remaining $a_i$'s will satisfy the set of constraints  $\{\inner{s,x}^2\le 1/8: s\in S\}$ and $T(x,x,x)\ge 1-1/\log n$. Therefore each time we can find a valid pseudo-expectation $\pE$. 

What we need to prove is for any pseudo-expectation $\pE$ we found, it always satisfies $\pE[\inner{a_i,x}^k]\ge e^{-\epsilon k}$ for some $k = O((\log n)/\epsilon)$ for some small enough constant $\epsilon$. Then by Theorem~\ref{thm:BKS}we can obtain a new vector that is $O(\epsilon)$-close to one of the $a_i$'s. We formalize this in the following lemma:

\begin{algorithm}\caption{Overcomplete Random 3-Tensor Decomposition}\label{alg:tensor-decomp}
	\begin{algorithmic}[1]
		\REQUIRE Random 3-tensor $T = \sum_{i=1}^m a_i^{\otimes 3}\sim \mathcal{D}_{m,n}$. %, parameters $\epsilon,\tau$. 
		\ENSURE $\hat{a}_1,\dots,\hat{a}_m \in \mathbb{R}^n$ s.t.  $\{\hat{a}_i\}$ is $0.1$-close to $\{a_i\}$
		\STATE $S\leftarrow \emptyset$
		\REPEAT 
			\STATE Using semidefinite programming to find a degree $k=O(\log n)$ pseudo-expectation $\pE$ that satisfies the constraints $\{T(x,x,x) \ge 1 - 1/\log n, \|x\|^2 = 1\}$ and $\{\inner{s,x}^2 \le 1/8: s\in S\}$. 
			\STATE Run the algorithm in Theorem 5.1 of~\cite{BKS14} (for $n^{O(k)}$ times) with input $\pE$ and obtain vector $c$ such that $T(c,c,c)\ge 0.99.$ %or sufficiently small constant $\tau$.  
			\STATE add vector $c$ to $S$. 
		\UNTIL{$|S| = m$}
		\RETURN $\{\hat{a}_i\} = S$. 
		%\STATE Solve the following optimization and obtain the optimal pseudo-distribution $\pEstar$
%		\begin{eqnarray}
%		\textrm{Maximize} && \pE \left[T\cdot x^{\otimes 3}\right] \nonumber\\
%		\textrm{Subject to } && \pE \textrm{ is a degree-12 pseudo-expectation}\\
%		&& \pE(r(x)g(x)) = 0 \textrm{ for } r(x) = \|x\|^2 - 1 \textrm{ and any $g$ with $\deg(g)\le 10$}
%		\end{eqnarray}
%		\STATE 
	\end{algorithmic}
\end{algorithm}

\begin{lemma} \label{lem:consistent}
When $T$ is chosen from $\mathcal{D}_{m,n}$ where $m\ll n^{3/2}$, with high probability over the randomness of $T$, the pseudo-expectation found in Step 3 of Algorithm~\ref{alg:tensor-decomp} satisfies the following: there exists an $a_i$ such that $\tilde{E}[\inner{a_i,x}^k]\ge e^{-\epsilon k}$ for sufficiently small constant $\epsilon$ (where the pseudo-expectation has degree  $4k$ and $k = O((\log n)/\epsilon)$). In particular, applying Theorem~\ref{thm:BKS}, repeat the algorithm for $n^{O(k)}$ time will give a vector $c$ such that $\inner{c,a_i} \ge 1-O(\epsilon)$.
\end{lemma}

The main intuition is to use Cauchy-Schwarz and H\"older inequalities (like what we used in Claim~\ref{clm:cs}) to raise the power in the sum $\sum_{i=1}^m \inner{a_i,x}^d$ (we start with $d=3$ and hope to get to $d = k$). When the degree is high enough we can afford to do an averaging argument and lose a factor of $m$ to go from the sum to a individual vector, because $e^{-\epsilon k} = \poly(m)$. The detailed proof is given in Appendix~\ref{sec:appendix:consistent}.
%We use the following polynomial program and its relaxation

Now we are ready to prove Theorem~\ref{thm:quasi-decomp-algo}. 

\begin{proof} (sketch)
We prove Theorem~\ref{thm:quasi-decomp-algo} by induction. Suppose $s$ already contains a set of vectors $\hat{a}_i$'s, where for each $\hat{a}_i$ there is a corresponding $a_j$ that satisfies $\|\hat{a}_i - a_j\|\le 0.1$. We would like to show with high probability in the next iteration, the algorithm finds a new component that is different from all the previously found $a_i$'s.

In order to do that, we need to show the following:
\begin{enumerate}
\item The SDP in Step 3 of Algorithm~\ref{alg:tensor-decomp} is feasible and gives a valid pseudo-expectation.
\item For any valid pseudo-expectation, with high probability we get an unit vector $c$ that satisfies $T(c,c,c) \ge 0.99$, and $c$ is far from all the previously found $a_i$'s.
\item For any unit vector $c$ such that $T(c,c,c)\ge 0.99$, there must be a component $a_i$ such that $\|a_i-c\| \le 0.1$.
\end{enumerate}

In these three steps, Step 1 follows because we can take $\pE$ to be the expectation of a true distribution: $x = a_i$ with probability $1$ for some unfound $a_i$. Step 2 is basically Lemma~\ref{lem:consistent}, when we choose $\epsilon$ to be a small enough constant, it is easy to prove that all the vectors that satisfy $\inner{c,a_i} \ge 1-O(\epsilon)$ must satisfy $T(c,c,c) \ge 0.99$. Step 3 is the second part of our observation in Section~\ref{sec:mainidea}, which we prove in the appendix.
\end{proof}

The details in this proof can be found in Appendix~\ref{sec:appendix:algmain}.

\section{Conclusion}
In this paper we give the first algorithm that can decompose an overcomplete 3rd order tensor when the rank $m$ is almost $n^{3/2}$ that matches the $n^{p/2}$ bounds for even order tensors. Our argument is based on a special unfolding of the tensor and a decoupling argument for matrix concentration. We feel such techniques can be useful in other settings.

Tensor decompositions are widely applied in machine learning for learning latent variable models. Although the SoS based algorithm have poor dependency on the accuracy $\epsilon$, in the case of tensor decomposition we can actually use SoS as an initialization algorithm. We hope such ideas can help solving more problems in machine learning.

\paragraph{Acknowledgment} We thank Anima Anandkumar, Boaz Barak, Johnathan Kelner, David Steurer, Venkatesan Guruswami for helpful discussions at various stages of this work.

\bibliography{ref}

\newcommand{\etalchar}[1]{$^{#1}$}
\begin{thebibliography}{BCMV14}

\bibitem[ABG{\etalchar{+}}13]{anderson2013more}
Joseph Anderson, Mikhail Belkin, Navin Goyal, Luis Rademacher, and James Voss.
\newblock The more, the merrier: the blessing of dimensionality for learning
  large gaussian mixtures.
\newblock {\em arXiv preprint arXiv:1311.2891}, 2013.

\bibitem[AFH{\etalchar{+}}13]{AnandkumarEtal:lda12}
A.~Anandkumar, D.~P. Foster, D.~Hsu, S.~M. Kakade, and Y.~K. Liu.
\newblock {Two SVDs Suffice: Spectral Decompositions for Probabilistic Topic
  Modeling and Latent Dirichlet Allocation}.
\newblock {\em to appear in the special issue of Algorithmica on New
  Theoretical Challenges in Machine Learning}, July 2013.

\bibitem[AFT11]{alexeev2011tensor}
Boris Alexeev, Michael~A Forbes, and Jacob Tsimerman.
\newblock Tensor rank: Some lower and upper bounds.
\newblock In {\em Computational Complexity (CCC), 2011 IEEE 26th Annual
  Conference on}, pages 283--291. IEEE, 2011.

\bibitem[AGH{\etalchar{+}}14]{AnandkumarEtal:tensor12}
A.~Anandkumar, R.~Ge, D.~Hsu, S.~M. Kakade, and M.~Telgarsky.
\newblock {Tensor Methods for Learning Latent Variable Models}.
\newblock {\em J. of Machine Learning Research}, 15:2773--2832, 2014.

\bibitem[AGHK13]{AnandkumarEtal:community12COLT}
A.~Anandkumar, R.~Ge, D.~Hsu, and S.~M. Kakade.
\newblock {A Tensor Spectral Approach to Learning Mixed Membership Community
  Models}.
\newblock In {\em Conference on Learning Theory (COLT)}, June 2013.

\bibitem[AGJ14]{AltTensorDecomp2014}
Anima Anandkumar, Rong Ge, and Majid Janzamin.
\newblock {Guaranteed Non-Orthogonal Tensor Decomposition via Alternating
  Rank-$1$ Updates}.
\newblock {\em arXiv preprint arXiv:1402.5180}, Feb. 2014.

\bibitem[AHK12]{AnandkumarHsuKakade:COLT12}
A.~Anandkumar, D.~Hsu, and S.~M. Kakade.
\newblock {A Method of Moments for Mixture Models and Hidden Markov Models}.
\newblock In {\em Proc. of Conf. on Learning Theory}, June 2012.

\bibitem[BBH{\etalchar{+}}12]{BBH12}
Boaz Barak, Fernando~G.S.L. Brandao, Aram~W. Harrow, Jonathan Kelner, David
  Steurer, and Yuan Zhou.
\newblock Hypercontractivity, sum-of-squares proofs, and their applications.
\newblock In {\em Proceedings of the Forty-fourth Annual ACM Symposium on
  Theory of Computing}, STOC '12, pages 307--326, New York, NY, USA, 2012. ACM.

\bibitem[BCMV14]{bhaskara2014smoothed}
Aditya Bhaskara, Moses Charikar, Ankur Moitra, and Aravindan Vijayaraghavan.
\newblock Smoothed analysis of tensor decompositions.
\newblock In {\em Proceedings of the 46th Annual ACM Symposium on Theory of
  Computing}, pages 594--603. ACM, 2014.

\bibitem[BKS14]{BarakKS14}
Boaz Barak, Jonathan~A. Kelner, and David Steurer.
\newblock Rounding sum-of-squares relaxations.
\newblock In {\em STOC}, pages 31--40, 2014.

\bibitem[BKS15]{BKS14}
Boaz Barak, Jonathan~A. Kelner, and David Steurer.
\newblock Dictionary learning and tensor decomposition via the sum-of-squares
  method.
\newblock In {\em Proceedings of the Forty-seventh Annual ACM Symposium on
  Theory of Computing}, STOC '15, 2015.

\bibitem[BM15]{BarakM15}
Boaz Barak and Ankur Moitra.
\newblock Tensor prediction, rademacher complexity and random 3-{XOR}.
\newblock 2015.

\bibitem[BS14]{BarakS14}
Boaz Barak and David Steurer.
\newblock Sum-of-squares proofs and the quest toward optimal algorithms.
\newblock In {\em Proceedings of International Congress of Mathematicians
  (ICM)}, 2014.
\newblock To appear.

\bibitem[Cha96]{Chang96}
Joseph~T. Chang.
\newblock Full reconstruction of {M}arkov models on evolutionary trees:
  Identifiability and consistency.
\newblock {\em Mathematical Biosciences}, 137:51--73, 1996.

\bibitem[Com14]{comontensorsurvey}
Pierre Comon.
\newblock Tensor: a partial survey.
\newblock {\em Signal Processing Magazine}, page~11, 2014.

\bibitem[DLCC07]{de2007fourth}
Lieven De~Lathauwer, Jos{\'e}phine Castaing, and Jean-Fran{\c{c}}ois Cardoso.
\newblock Fourth-order cumulant-based blind identification of underdetermined
  mixtures.
\newblock {\em Signal Processing, IEEE Transactions on}, 55(6):2965--2973,
  2007.

\bibitem[GHK15]{GHK15}
Rong Ge, Qingqing Huang, and Sham~M. Kakade.
\newblock Learning mixtures of gaussians in high dimensions.
\newblock In {\em Proceedings of the Forty-seventh Annual ACM Symposium on
  Theory of Computing}, STOC '15, 2015.

\bibitem[Gur03]{Gurvits}
Leonid Gurvits.
\newblock Classical deterministic complexity of edmonds' problem and quantum
  entanglement.
\newblock In {\em Proceedings of the Thirty-fifth Annual ACM Symposium on
  Theory of Computing}, STOC '03, pages 10--19, New York, NY, USA, 2003. ACM.

\bibitem[H{\aa}s90]{haastad1990tensor}
Johan H{\aa}stad.
\newblock Tensor rank is np-complete.
\newblock {\em Journal of Algorithms}, 11(4):644--654, 1990.

\bibitem[HL09]{TensorNPHard}
Christopher~J. Hillar and Lek-Heng Lim.
\newblock {Most tensor problems are NP hard}.
\newblock {\em arXiv preprint arXiv:0911.1393}, 2009.

\bibitem[HM13]{harrow2013testing}
Aram~W Harrow and Ashley Montanaro.
\newblock Testing product states, quantum merlin-arthur games and tensor
  optimization.
\newblock {\em Journal of the ACM (JACM)}, 60(1):3, 2013.

\bibitem[Kru77]{Kruskal:77}
J.B. Kruskal.
\newblock {Three-way arrays: Rank and uniqueness of trilinear decompositions,
  with application to arithmetic complexity and statistics}.
\newblock {\em Linear algebra and its applications}, 18(2):95--138, 1977.

\bibitem[Las01]{lasserre2001global}
Jean~B Lasserre.
\newblock Global optimization with polynomials and the problem of moments.
\newblock {\em SIAM Journal on Optimization}, 11(3):796--817, 2001.

\bibitem[MR06]{MR06}
Elchanan Mossel and S\'{e}bastian Roch.
\newblock Learning nonsingular phylogenies and hidden {M}arkov models.
\newblock {\em Annals of Applied Probability}, 16(2):583--614, 2006.

\bibitem[Par00]{parrilo2000structured}
Pablo~A Parrilo.
\newblock {\em Structured semidefinite programs and semialgebraic geometry
  methods in robustness and optimization}.
\newblock PhD thesis, California Institute of Technology, 2000.

\bibitem[PMS95]{decoupling}
Victor H. de~la Pena and S.~J. Montgomery-Smith.
\newblock Decoupling inequalities for the tail probabilities of multivariate
  u-statistics.
\newblock {\em The Annals of Probability}, 23(2):pp. 806--816, 1995.

\bibitem[Str73]{strassen1973vermeidung}
Volker Strassen.
\newblock Vermeidung von divisionen.
\newblock {\em Journal f{\"u}r die reine und angewandte Mathematik},
  264:184--202, 1973.

\bibitem[Tro12]{tropp2012user}
Joel~A Tropp.
\newblock User-friendly tail bounds for sums of random matrices.
\newblock {\em Foundations of Computational Mathematics}, 12(4):389--434, 2012.

\end{thebibliography}
\bibliographystyle{alpha}

\appendix
\newpage

\section{Omitted Proofs in Section~\ref{sec:certify}}

\subsection{Proof of Lemma~\ref{lem:2-4norm}}\label{sec:2-4norm}
\label{sec:append:2to4}
We first restate the lemma here.

\begin{lemma}With high probability over the randomness of $a_i$'s, 
	\begin{equation}
	\sum_{i=1}^m \inner{a_i,x}^4\sosle_{r,12} 1+\widetilde{O}(m/n^{3/2}) \label{eqn:inter10}
	\end{equation}
\end{lemma}

%We are going to have an almost rigorous SoS proof. The slight non-rigorousness comes from the fact that we use some kind of recursive argument. 
%As in the $2\rightarrow 3$ norm case, we first use Chauchy-Shwarz to raise it into higher power: 

Recall \cite{BBH12} showed that when $m \ll n^2$, 

\begin{equation}
\sum_{i=1}^m \inner{a_i,x}^4\le O(1) \label{eqn:weak-2-4-norm}
\end{equation}

Here in order to improve this bound, we consider the square of the LHS of~\eqref{eqn:2-4norm} and apply Cauchy-Schwarz (similar to Claim~\ref{clm:cs}), 
\begin{align}
\left(\sum_{i=1}^m \inner{a_i,x}^4\right)^2 &= \left\langle\sum_{i=1}^m\inner{a_i,x}^3a_i, x\right\rangle^2\nonumber \\
&\sosle \left\|\sum_{i=1}^m\inner{a_i,x}^3a_i\right\|^2\|x\|^2 &\textrm{by Cauchy-Schwarz}\nonumber \\
&\sosle_r\left\|\sum_{i=1}^m\inner{a_i,x}^3a_i\right\|^2 = \sum_{i=1}^m \inner{a_i,x}^6 + \sum_{i\neq j}\inner{a_i,a_j}\inner{a_i,x}^3\inner{a_j,x}^3\label{eqn:two-terms-2}
\end{align}

We will bound the first term of~\eqref{eqn:two-terms-2} by $1+o(1)$. We simply let $y = x^{\otimes 3}$ and let $B$ be the matrix whose $i$th row is $a_i^{\otimes 3}$. Then $f(y) = \|By\|^2$ has the property that $f(x^{\otimes 3}) =  \sum_{i=1}^m \inner{a_i,x}^6$. Therefore it suffices to prove that $f(y)\sosle (1 + o(1)\|y\|^2$ or equivalently $\|B\|\le 1+o(1)$. %Then we can rewrite 

%$$\sum_{i=1}^m \inner{a_i,x}^6 = \|By\|^2$$

Consider the matrix $BB^T$. It is a $n$ by $n$ matrix with diagonal entries 1 and off diagonal entries of the form $\inner{a_i^{\otimes 3},a_j^{\otimes 3}} = \inner{a_i,a_j}^3$. 
By the incoherence of $a_i$'s, we have $\inner{a_i,a_j}^3 \lesssim 1/n^{3/2}$. Then by Gershgorin disk theorem, we have $\|BB^T\|\le 1+\widetilde{O}(m/n^{3/2})  = 1+\delta$. It follows that $\|B\|\le 1+\widetilde{O}(m/n^{3/2})$.  Therefore, 
\begin{equation}
\sum_{i=1}^m \inner{a_i,x}^6 = \|Bx^{\otimes 3}\|^2 \sosle (1+\widetilde{O}(m/n^{3/2}))\|x^{\otimes 3}\| \le_r 1+\widetilde{O}(m/n^{3/2}) \label{eqn:inter6}
\end{equation}

For the second term of~\eqref{eqn:two-terms-2}, we apply Cauchy-Schwarz again: 

\begin{align}
\left(\sum_{i\neq j}\inner{a_i,a_j}\inner{a_i,x}^3\inner{a_j,x}^3\right)^2 &\sosle   \left(\sum_{i\neq j}\inner{a_i,a_j}^2\inner{a_i,x}^2\inner{a_j,x}^2\right) \left(\sum_{i\neq j}\inner{a_i,x}^4\inner{a_j,x}^4\right) \nonumber\\
&\sosle \left(\frac{1}{n}\cdot \sum_{i}\inner{a_i,x}^2\sum_j \inner{a_j,x}^2\right) \left(\sum_{i}\inner{a_i,x}^4\sum_j \inner{a_j,x}^4\right) \label{eqn:inter4}
\end{align}

Note that the matrix $A =[a_1\vert \dots\vert a_m]$ has spectral norm bound  $\|A\|\lesssim \sqrt{m/n}$, and therefore 
$$\sum_{i}\inner{a_i,x}^2 = \|A^Tx\|^2 \sosle \|A\|^2 \|x\|^2 \sosle_r \|A\|^2$$
Then using Equation~\ref{eqn:weak-2-4-norm}, and the equation above, we have %the spectral norm of $[a_1\vert \dots\vert a_m]$, we have 
\begin{align}
%\left(\sum_{i\neq j}\inner{a_i,a_j}\inner{a_i,x}^3\inner{a_j,x}^3\right)^2 
%&\le \left(\frac{1}{n}\cdot \sum_{i}\inner{a_i,x}^2\sum_j \inner{a_j,x}^2\right) \left(\sum_{i}\inner{a_i,x}^4\sum_j \inner{a_j,x}^4\right)\\
\textrm{RHS of~\eqref{eqn:inter4} }& \sosle_r \frac{1}{n}\cdot \frac{m}{n} \cdot \frac{m}{n} \cdot O(1) \cdot O(1) \le \widetilde{O}(m^2/n^3) \label{eqn:inter5}
\end{align}

Then by ~\ref{eqn:inter4} and ~\ref{eqn:inter5} and Lemma~\ref{lem:square-root}, we have that 
\begin{equation}
\sum_{i\neq j}\inner{a_i,a_j}\inner{a_i,x}^3\inner{a_j,x}^3 \sosle_r \widetilde{O}(m^2/n^3) \label{eqn:inter7}
\end{equation}

Hence, combining equation~\eqref{eqn:inter7}, ~\eqref{eqn:inter6} and~\eqref{eqn:two-terms-2} we have that

\begin{align}
\left(\sum_{i=1}^m \inner{a_i,x}^4\right)^2 
&\sosle_r \sum_{i=1}^m \inner{a_i,x}^6 + \sum_{i\neq j}\inner{a_i,a_j}\inner{a_i,x}^3\inner{a_j,x}^3\label{eqn:inter11}\\
& \sosle_r 1+ \widetilde{O}(m/n^{3/2})  + \widetilde{O}(m/n^{3/2})   = 1+\widetilde{O}(m/n^{3/2}) \nonumber
\end{align}
Using Lemma~\ref{lem:square-root} again, we complete the proof of Lemma~\ref{eqn:2-4norm}. 

\subsection{Proof of Lemma~\ref{lem:px}}\label{sec:px}

We first restate the lemma:

\begin{lemma}
	When $m\ll n^{3/2}$, the matrix $M = \sum_{i\neq j}\inner{a_i,a_j} (a_i\otimes a_j) (a_i\otimes a_j)^T$ has spectral norm at most $\widetilde{O}(m/n^{3/2})$ and as a direct consequence, 
	$$p(x)\sosle_{r,4} \widetilde{O}(m/n^{3/2})$$
\end{lemma}

\begin{proof}
	As suggested in the proof sketch, we first use a simple symmetrization which allows us to focus on the randomness of signs of $\inner{a_i,a_j}$.  For simplicity of notation, let $Q_{ij} := \inner{a_i,a_j}(a_i\otimes a_j) (a_i\otimes a_j)^T$.  Let $\sigma\in \{\pm 1\}^m$ be uniform random $\pm 1$ vector and define $M'$ as 
	$$M' =  \sum_{i\neq j} \sigma_i\sigma_jQ_{ij}.$$
	We claim that $M'$ has the same distribution as $M$, since $a_i$ has the same distribution as $\sigma_ia_i$. Then from now on we condition on the event that $a_i$'s have incoherence property and low spectral norm, that is, $\inner{a_i,a_j}\lesssim 1/\sqrt{n}$,   $\|A\| = \left\|[a_1\vert a_2 \dots \vert a_m]\right\|\lesssim \sqrt{m/n}$, and we will only focus on the randomness of $\sigma$. Ideally we want to write $M'$ as a sum of independent random matrices so that we can apply matrix Bernstein inequality. However, now the random coefficients are $\sigma_i\sigma_j$, and they are not independent with each other.
	
	A key observation here is that the sum is only over the indices $(i,j)$ with $i\neq j$, therefore we can use Theorem 1 of~\cite{decoupling} (restated as Theorem~\ref{thm:decoupling} in the end) to decouple the correlation first. 
	
	Theorem~\ref*{thm:decoupling} basically says that to study the concentration of  a sum of the form $\sum_{i\neq j}f_{ij}(X_i,X_j)$, it is up to constant factor similar to the concentration of the sum $\sum_{i\neq j}f_{ij}(X_i,Y_j)$ where $Y_i$ is an independent copy of $X_i$. Applying the theorem to our situation, we have that there exists absolute constant $C$ such that 
	\begin{equation}
	\Pr[\|M'\|\ge t] \le C\Pr[M''\ge t/C] \label{eqn:M'-M''}
	\end{equation}
	where $$M'':= \sum_{i\neq j}\sigma_i\tau_jQ_{ij},$$ and $\sigma,\tau$ are independently uniform over $\{-1,+1\}^m$. 
	
	%(and they are independent of each other as well).
	Now it suffices to bound the norm of $M''$. We proceed by rewriting $M''$ as 
	$$M'' = \sum_i {\sigma_i \sum_{j\neq i}\tau_jQ_{ij}} := \sum_i \sigma_i T_i,$$
	where \begin{equation}
	T_i :=  \sum_{j\neq i}\tau_jQ_{ij} \label{eqn:T_i}
	\end{equation}
	
	We study the properties of $T_i$ first. 
	
	\begin{claim}\label{claim:T}
		With high probability over the randomness of $a_i$'s, for all $i$, $T_i \preceq \tilde{O}(\sqrt{m}/n) (a_i a_i^T) \otimes I$.
	\end{claim}
	
	\begin{proof}	
		
		Recall that $Q_{ij} =  \inner{a_i,a_j}(a_i\otimes a_j) (a_i\otimes a_j^T)$. In the definition~\ref{eqn:T_i} of $T_i$, the index $i$ is fixed and we take sum over $j$. Therefore it will be convenient to write $Q_{ij}$ as $Q_{ij} = \inner{a_i,a_j}(a_ia_i^T)\otimes (a_ja_j)^T$ where $\otimes$ is the Kronecker product between matrices. Then $T_i$ can be written as 
		
		$$T_i = (a_ia_i^T)\otimes \left(\sum_j\tau_j\inner{a_i,a_j}a_ja_j^T\right).$$
		
		We apply the Matrix Bernstein inequality (Theorem~\ref{thm:matbernstein}) on the right factor. Matrix Bernstein bound requires spectral norm bound for individual matrices, and a variance bound.
		
		For the spectral norm of individual matrices, we check that $ \|\tau_j\inner{a_i,a_j}a_ja_j^T\|\lesssim 1/\sqrt{n}$ (by incoherence).
		For variance we know $$\|\Exp[\sum_{j}\tau_j^2(\inner{a_i,a_j}a_ja_j^T)^2]\| = \|A\diag(\inner{a_i,a_j}^2)_{j\neq i}A^T\|\lesssim m/n^{2},$$ 
		where we used the spectral norm of $A$ and the fact that $\inner{a_i,a_j}^2 \lesssim 1/n$.
		
		Therefore by Matrix Bernstein's inequality (Theorem~\ref{thm:matbernstein}) we have that  whp, over the randomness of $\tau$, 
		$$\|\sum_j\tau_j\inner{a_i,a_j}a_ja_j^T\| \le \widetilde{O}(\sqrt{m}/n).$$
		
		Using the fact that for two matrices $P$ and $Q$, if $P\preceq Q$ and $R$ is PSD, then $R\otimes P\preceq R\otimes Q$ (see Claim~\ref{lem:psdkronecker}),  it follows that 
		$$T_i \preceq (a_ia_i^T)\otimes (\widetilde{O}(\sqrt{m}/n)\cdot I).$$
		
		Finally we use union bound and conclude with high probability this is true for any $i$.
	\end{proof}
	
	%	Then we have $$M''  = \sum_i \sigma_i T_i \le \left(\sum_i \sigma_i a_ia_i^T\right)\otimes (\sqrt{m}/n\cdot I)$$
	
	Now we can apply matrix Bernstein for the sum $M'' = \sum_{i=1}^m \sigma_i T_i$. The individual spectral norm is bounded by $\tilde{O}(\sqrt{m}/n)$ by the Claim~\ref{claim:T}. The variance is
	$$
	\|\sum_{i=1}^m T_i^2\| \le \tilde{O}(m/n^2)\|\sum_{i=1}^m ((a_ia_i^T)\otimes I)^2\|
	= \tilde{O}(m/n^2)\|(AA^T)\otimes I\| = \tilde{O}(m^2/n^3).
	$$
	
	Using matrix Bernstein inequality, we know with high probability $\|M''\| \le \tilde{O}(m/n^{3/2})$.
	
	Using (\ref{eqn:M'-M''}), we get that whp, $\|M'\|\le \widetilde{O}(m/n^{3/2})$. Since $M'$ and $M$ has the same distribution, we conclude that  whp, $\|M\|\le \widetilde{O}(m/n^{3/2})$. 
\end{proof}

We complete the proof by providing the following claim about Kronecker products.

\begin{claim}
	\label{lem:psdkronecker}
	If $P\preceq Q$ and $R$ is psd, then $R\otimes P\preceq R\otimes Q$.
\end{claim}

\begin{proof}
	It suffices to prove this when $R = uu^T$ (as we can always decompose $R$ as sum of rank one components). In that case, for any $y \in \R^{n^2}$, we can write $y = u\otimes v + z$ where $z$ is orthogonal to $u\otimes e_i$ for all $i\in[n]$. Now $(R\otimes P) z = 0$, therefore
	$$
	y^T (R\otimes P)y = (u\otimes v)^T(R\otimes P)(u\otimes v) = (u^TRu) (v^TPv) \le (u^TRu)(v^TQv) = y^T (R\otimes Q)y.$$
	Therefore $R\otimes P\preceq R\otimes Q$.
\end{proof}

\subsection{Main Theorem for Certifying Injective Norm}
\label{sec:append:certifymain}
Now we are ready to prove Theorem~\ref{thm:maininjdet}.

\begin{theorem}
	Algorithm~\ref{alg:certifying-norm}  always returns NO when $\|T\|_{\textrm{inj}} > 1+1/\log n$. When $T\sim \mathcal{D}_{m,n}$ and $m\ll n^{3/2}$, Algorithm~\ref{alg:certifying-norm} returns YES with high probability over the randomness of $T$. Further, the same guarantee holds given an approximation $\tilde{T}$ where if $M \in \R^{n\times n^2}$ is an unfolding of $T - \tilde{T}$, $\|M\| \le 1/2\log n$. 
\end{theorem}

\begin{proof}
	We first prove whenever $\|T\|_{\textrm{inj}} > 1+1/\log n$, the algorithm returns NO. This is because a large injective norm implies there exists an unit vector $x^*$ with $T(x^*,x^*,x^*) = 1$. We can construct a pseudo-expectation $\pE$ as $\pE[p(x)] = p(x^*)$. Clearly this is a valid pseudo-expectation (it is even the expectation of a true distribution: $x = x^*$ with probability 1). Also, we know $\pE[T(x,x,x)] = T(x^*,x^*,x^*) > 1+1/\log n$, so in particular $OPT > 1+1/\log n$ and the algorithm must return NO.
	
	Next we show the algorithm returns YES with high probability when $T$ is chosen from $\mathcal{D}$. This follows directly from Theorem~\ref{thm:proof}, which in turn follows from Lemmas~\ref{lem:2-4norm} and \ref{lem:px}. In particular, we know there is a degree-12 SoS proof that shows $T(x,x,x) \le 1+\tilde{O}(m/n^{3/2}) \le 1+1/2\log n$, so by Lemma~\ref{lem:sos} this must also hold for any pseudo-expectation.
	
	When we are only given tensor $\tilde{T}$ such that the unfolding of $\tilde{T}-T$ has spectral norm $1/2\log n$. Let $M$ be the unfolding of $\tilde{T}-T$, and $y = x\otimes x$, then by Lemma~\ref{lem:sosspectral} we know $(x^TMy)^2 \preceq \|x\|^2 \|M\|^2 \|y\|^2$, which implies (by Lemma~\ref{lem:sosspectral}) $[\tilde{T}-T](x,x,x) = x^TMy \sosle_{r,12} \|M\| \le 1/2\log n$. Combining the two terms we know
	$$
	\tilde{T}(x,x,x) = T(x,x,x)+[\tilde{T}-T](x,x,x) \sosle_{r,12} 1+1/\log n.
	$$
\end{proof}

\section{Omitted Proof in Section~\ref{sec:alg}}

\subsection{Proof of Lemma~\ref{lem:consistent}}

\label{sec:appendix:consistent}

We first restate the lemma here:

\begin{lemma}
When $T$ is chosen from $\mathcal{D}_{m,n}$ where $m\ll n^{3/2}$, with high probability over the randomness of $T$, the pseudo-expectation found in Step 3 of Algorithm~\ref{alg:tensor-decomp} satisfies the following: there exists an $a_i$ such that $\tilde{E}[\inner{a_i,x}^k]\ge e^{-\epsilon k}$ for sufficiently small constant $\epsilon$ (where the pseudo-expectation has degree  $4k$ and $k = O((\log n)/\epsilon)$). In particular, applying Theorem~\ref{thm:BKS}, repeat the algorithm for $n^{O(k)}$ time will give a vector $c$ such that $\inner{c,a_i} \ge 1-O(\epsilon)$.
\end{lemma}

First we will show that for a valid pseudo-expectation, the sum of $\inner{a_i,x}^4$ and $\inner{a_i,x}^6$ are also bounded. This actually follows directly from the proof of Lemma~\ref{lem:2-4norm} and~\ref{lem:px}. 

%The following lemma is simply the pseudo-expectation version of Lemma~\ref{lem:2-4norm} and Lemma~\ref{lem:px}. The proof follows straightforwardly from taking pseudo-expectation over the proofs of Lemma~\ref{lem:2-4norm} and~\ref{lem:px}. 
\begin{lemma}
	With high probability over the randomness of $T$, we have that for any degree-12 pseudo expectation $\pE$ that satisfies the constraints $\{\|x\|^2 = 1, T(x,x,x) \ge 1-\tau\}$, it also satisfies 
	
	\begin{equation}
	1 + \epsilon\ge \pE\left[\sum_{i=1}^m \inner{a_i,x}^4\right] \ge 1 - \epsilon\label{eqn:quasi-4-lower-bound}
	\end{equation}
		\begin{equation}
		1 + \epsilon \ge \pE\left[\sum_{i=1}^m \inner{a_i,x}^6\right] \ge 1 - \epsilon \label{eqn:quasi-6-lower-bound}
		\end{equation}
	for $\epsilon = \widetilde{O}(m/n^{3/2}) + O(\tau)$.
\end{lemma}

\begin{proof}
	We essentially just take pseudo-expectation on the SoS proofs for Lemma~\ref{lem:2-4norm} and~\ref{lem:px}. The upper bounds follows directly by taking pseudo-expectation on equation~\eqref{eqn:inter10} and~\eqref{eqn:inter6}. Fo the lower bounds, by taking pseudo-expectation over the SoS equation in Lemma~\ref{lem:px}, we have that $\pE\left[p(x)\right] \le \widetilde{O}(m/n^{3/2})$. Taking pseudo-expectation over Claim~\ref{clm:cs}, using the assumption that $\pE$ satisfies $T(x,x,x)\ge 1-\tau$,  we have that 
	\begin{equation}
	1-\tau\le \pE\left[[T(x,x,x)]^2\right] \le \pE\left[\inner{a_i,x}^4\right] + \pE\left[p(x)\right] \le \pE\left[\inner{a_i,x}^4\right] + \widetilde{O}(m/n^{3/2}) %\label{eqn:two-terms}
	\end{equation}
	which implies 
	\begin{equation}
	\pE\left[\inner{a_i,x}^4\right] \ge 1-\tau-\widetilde{O}(m/n^{3/2}).\label{eqn:inter12}
	\end{equation}

	For proving the lower bounds in~\eqref{eqn:quasi-6-lower-bound}, we first  pseudo-expectation on  equation~\ref{eqn:inter7}, we have that 
	\begin{equation}
	\pE\left[\sum_{i\neq j}\inner{a_i,a_j}\inner{a_i,x}^3\inner{a_j,x}^3 \right]\le \widetilde{O}(m^2/n^3) \nonumber%\label{eqn:inter7}
	\end{equation}
	
	Then taking pseudo-expectation over equation~\eqref{eqn:inter11}, we obtain that 
	\begin{align}
	\pE\left[\left(\sum_{i=1}^m \inner{a_i,x}^4\right)^2\right] 
	\le \pE\left[\sum_{i=1}^m \inner{a_i,x}^6\right] + \pE\left[\sum_{i\neq j}\inner{a_i,a_j}\inner{a_i,x}^3\inner{a_j,x}^3\right]\nonumber %\\
	%& \le  1+ \widetilde{O}(m/n^{3/2})  + \widetilde{O}(m/n^{3/2})   = 1+\widetilde{O}(m/n^{3/2}) %\label{eqn:inter11}
	\end{align}
	
	Note that by equation~\eqref{eqn:inter12} and Cauchy-Schwarz, we have $$\pE\left[\left(\sum_{i=1}^m \inner{a_i,x}^4\right)^2\right]  \ge \left(\pE\left[\sum_{i=1}^m \inner{a_i,x}^4\right] \right)^2\ge 1-O(\tau)-\widetilde{O}(m/n^{3/2})$$
	Combining the two equations above, we obtain that 
	
	$$\pE\left[\sum_{i=1}^m \inner{a_i,x}^6\right] \ge 1-O(\tau)-\widetilde{O}(m/n^{3/2})$$
%{\sc RG: We still need to say something, just write the Cauchy-Schwartz again and explain that all the terms are bounded before and we just need to take pseudo-expectation.}
\end{proof}

Next we are going to prove that $\pE$ also satisfies the condition of Theorem 5.2 of~\cite{BKS14}.  %for the degree is $O(\log m)$. 
\begin{lemma}\label{lem:condtion-for-theorem5.1}
	For $k = O((\log n)/\epsilon)$ with constant $\epsilon < 1$, If $\pE$ is a degree-$k$ pseudo-expectation  that satisfies equation~\eqref{eqn:quasi-6-lower-bound} and~\eqref{eqn:quasi-4-lower-bound}, then there must exists $i\in [m]$ such that $\pE[\inner{a_i,x}^k]\ge e^{-(2\epsilon +\delta)k}$ with $\delta = \widetilde{O}(m/n^{3/2})$.
\end{lemma}
\begin{proof}
	By equation (2.5) of~\cite{BKS14}, we the following SoS version of Holder inequality. For any integer $t,d$ and $k = t(d-2)$, 
	 $$\|v\|_{d}^{dt} \sosle_{k} \|v\|^k_k\cdot \|v\|^{2t}$$ %can be proved via SoS for $k = t(d-2)$ and any vector $v$. 
	Let $v_i = \inner{a_i,x}^2$, we have 
	\begin{equation}
	\left(\sum_{i=1}^m \inner{a_i,x}^{2d}\right)^t \sosle_{k} \sum_{i=1}^m \inner{a_i,x}^{2k}\cdot \left(\sum_{i=1}^m \inner{a_i,x}^{4}\right)^t\label{eqn:quasi-1}
	\end{equation}
	By Lemma~\ref{lem:2-4norm}, we have that with high probability over randomness of $a_i$'s, 	$\sum_{i=1}^m \inner{a_i,x}^{4}\sosle 1+\widetilde{O}(m/n^{3/2}) $, and it follows that 
	\begin{equation}
	\left(\sum_{i=1}^m \inner{a_i,x}^{4}\right)^t \le (1+  \widetilde{O}(m/n^{3/2}))^t \label{eqn:quasi-2}
	\end{equation}
	By picking $d = 3$, we have $t = k$. Taking $t = O(\log m/\epsilon)$ and combining equation~\eqref{eqn:quasi-1} and~\eqref{eqn:quasi-2}, we have that %and by the constraints $\sum_{i=1}^m \inner{a_i,x}^6 \ge 1-o(1)$, we have that 
		\begin{equation}
		\left(\sum_{i=1}^m \inner{a_i,x}^{6}\right)^k \sosle_{k} \sum_{i=1}^m \inner{a_i,x}^{2k}\cdot \left(\sum_{i=1}^m \inner{a_i,x}^{4}\right)^k \sosle_k(1+  \widetilde{O}(m/n^{3/2}))^k\sum_{i=1}^m \inner{a_i,x}^{2k} \nonumber
		\end{equation}
	Applying pseudo-expectation on both hands, we obtain, 
			\begin{equation}
			\pE\left[\left(\sum_{i=1}^m \inner{a_i,x}^{6}\right)^k\right] \le (1+  \widetilde{O}(m/n^{3/2}))^k\cdot \pE\left[\sum_{i=1}^m \inner{a_i,x}^{2k} \right]\nonumber
			\end{equation}
	Note that by Cauchy-Schwarz and equation~\eqref{eqn:quasi-6-lower-bound}, we have 
				\begin{equation}
				(1-\epsilon)^k\le \pE\left[\sum_{i=1}^m \inner{a_i,x}^{6}\right]^k \le\pE\left[\left(\sum_{i=1}^m \inner{a_i,x}^{6}\right)^k\right]\nonumber
				\end{equation}
				
	Combining the two equations above, we obtain that 	for $\delta = \widetilde{O}(m/n^{3/2})$, 
	
	\begin{equation}
	\pE\left[\sum_{i=1}^m \inner{a_i,x}^{2k}\right] \ge (1-\delta)^k (1-\epsilon)^k
	\end{equation}
%	$$\left(\sum_{i=1}^m \inner{a_i,x}^{2d}\right)^t \ge (1-o(1))^t$$
	
%	Note that when $d=3$, wwe have $k = t$ and therefore, we can lower bound $\sum_{i=1}^m \inner{a_i,x}^{2k}$  via SoS by  
	
	%$$\sum_{i=1}^m \inner{a_i,x}^{2k} \ge (1-\delta)^k$$

	Therefore by averaging argument, there exists $i$ such that 
	$$\pE[\inner{a_i,x}^{2k}]\ge (1-\delta)^k/m =  e^{-\delta k - \log m -\epsilon k }$$
	
	when $ k\ge (\log m)/\epsilon$, we have that $\pE[\inner{a_i,x}^{2k}]\ge e^{-(2\epsilon+\delta)k}$
\end{proof}

Lemma~\ref{lem:consistent} follows directly from the two lemmas above.

\subsection{Proof of Theorem~\ref{thm:quasi-decomp-algo}}

\label{sec:appendix:algmain}

In this section we prove the main theorem in Section~\ref{sec:alg}.

\begin{theorem}
Let $T$ be a tensor chosen from $\mathcal{D}_{m,n}$, when $m\ll n^{3/2}$ with high probability over the randomness of $T$  Algorithm~\ref{alg:tensor-decomp}	returns $\{\hat{a}_i\}$ that is $0.1$-close to $\{a_i\}$ in time $n^{O(\log n)}$. Further, the same guarantee holds given an approximation $\tilde{T}$ where if $M \in \R^{n\times n^2}$ is an unfolding of $T - \tilde{T}$, $\|M\| \le \sfrac{1}{10\log }n$. 
\end{theorem}

As suggested in the proof sketch, we prove this theorem by induction. The induction hypothesis is that all vectors $s_i\in S$ are 0.1-close (in $\ell_2$ norm) to distinct components $a_i$'s. We break the proof into three claims:

\begin{claim}
With high probability over the tensor $T$, suppose all the previously found $s_i$'s are $0.1$-close (in $\ell_2$ norm) to some components $a_j$'s, then there exists a pseudo-expectation that satisfies Step 3 in Algorithm~\ref{alg:tensor-decomp}.
\end{claim}

\begin{proof}
We first prove that with high probability $T(a_i,a_i,a_i) \ge 1-1/\log n$ for all $i$. This is easy because $T(a_i,a_i,a_i) = 1 + \sum_{j\ne i} \inner{a_i,a_j}^3$. Conditioned on $a_i$, the values $\inner{a_i,a_j}$ are sub-Gaussian random variables with mean 0 and variance $1/n$, so by standard concentration bounds we know with high probability $\sum_{j\ne i} \inner{a_i,a_j}^3 \ge -1/\log n$. We can then take the union bound and conclude $T(a_i,a_i,a_i) \ge 1-1/\log n$ for all $i$.

Now for simplicity of notation, assume that $S = \{s_1,\dots,s_t\}$ for some $t < m$,  where $s_i$ is $0.1$-close to $a_i$. We can construct a pseudo-expectation $\pE[p(x)] = p(a_{t+1})$. Clearly this is a valid pseudo-expectation that satisfies $\|x\|^2 = 1$. For the inequality constraints we also know $\inner{a_{t+1},s_i}^2 \le 2(\inner{a_{t+1},a_i}^2+\inner{a_{t+1}, a_i-s_i}^2) < 1/8$ (where the whole proof only uses Cauchy-Schwarz and $(A+B)^2\le 2(A^2+B^2)$, so the proof is SoS). Therefore the system in Step 3 must have a feasible solution.
\end{proof}

\begin{claim}
For any valid pseudo-expectation in Step 3, with high probability we get an unit vector $c$ that satisfies $T(c,c,c) \ge 0.99$, and $c$ is far from all the previously found $a_i$'s.
\end{claim}

\begin{proof}
	By Lemma~\ref{lem:consistent} we know there must be a vector $a_i$ such that $\pE[\inner{a_i,x}^k] \ge e^{-\epsilon k}$ for sufficiently small constant $\epsilon$. We show that this vector $a_i$ cannot be among the previously found ones. By Lemma~\ref{lem:sos-holder} we know that for even number $k$, 
	$$(\inner{s_i,x}+\inner{s_i-a_i,x})^{k} \le 2^{k-1}(\inner{s_i-a_i,x}^{k}+\inner{s_i,x}^{k})$$
	Taking pseudo-expectations over both sides, we have that
	$$\pE[\inner{a_i,x}^k] \sosle_{2k} 2^{k-1}(\pE[\inner{s_i,x}^k] + k\pE[\inner{s_i-a_i,x}^k]) \sosle_{\|x|^2=1,2k} e^{-\epsilon k}$$ where we've used the constraint $\inner{s_i,x}^2\le 1/8$ and induction hypothesis $\|s_i-a_i\| \le 0.1$.
	%{\sc RG: we may need the SoS proof for this fact $(A+B)^k \le 2^{k-1}(A^k+B^k)$. This is probably just some version of Holder and if not the case when $k$ is a power of 2 is very easy to prove.}
	
	Now applying Theorem~\ref{thm:BKS} we get a vector $c$ that is has inner-product $1-O(\epsilon)$ with $a_i$. Therefore $T(c,c,c) = T(a_i,a_i,a_i) + T(c-a_i,a_i,a_i)+T(c,c-a_i,a_i) + T(c,c,a_i) \ge 1-1/\log n - 3\|T\|_{\textrm{inj}}\|c-a_i\| \ge 0.99$. Here $T(x,y,z) = \sum_{i_1,i_2,i_3}T_{i_1,i_2,i_3}x_{i_1}y_{i_2}z_{i_3}$ is the multilinear form for the tensor, and note that this step of the proof does not need to be SoS because we already have the vector $c$ from Theorem~\ref{thm:BKS}.
\end{proof}

\begin{claim}
For any unit vector $c$ such that $T(c,c,c)\ge 0.99$, there must be a component $a_i$ such that $\|a_i-c\| \le 0.1$.
\end{claim}

\begin{proof}
	We define the following trivial pseudo-expectation $\pE^c$ defined by $c$: $\pE^c\left[p(x)\right] =  p(c)$.  Then we know that $\pE^c$ does satisfy equation $T(x,x,x)\ge 0.99$, and the degree of $\pE^c$ can be any finite number. Therefore, by Lemma~\ref{lem:condtion-for-theorem5.1}, we have that $\pE^c\left[\inner{a_i,x}^k\right]\ge e^{-(2\epsilon+\delta
		)k}$ for $k = O(\log n)$. 
	Therefore using the definition of $\pE^c$, we have that 
	$\pE^c\left[\inner{a_i,x}^k\right] = \inner{a_i,c}^k \ge e^{-(2\epsilon+\delta
		)k}$. Taking $\epsilon = 0.001$ and then we have that $\inner{a_i,c} \ge 0.999-\delta$ and it follows that $\|a_i-c\|\le 0.99$. 
%{\sc RG: Please fill this in.}
\end{proof}

These three claims finishes the induction in the noiseless case. For the noisy case, we can handle it the same ways as Theorem~\ref{thm:maininjdet}: note that $[\tilde{T}-T](x,x,x) \sosle_{\|x\|^2=1,12} 1/2\log n$ and this additional term does not change any part of the proof.

Finally, the runtime of Line 3 in Algorithm~\ref{alg:tensor-decomp} is $n^{O(k)}$, and the run-time of line 4 is also $n^{O(k)}$. Therefore the total runtime is $n^{O(k)}$. 
\section{Matrix Concentrations}

In this section we introduce theorems used to prove matrix concentrations. First we need the following lemma for decoupling the randomness in the sum.

\begin{theorem}[Special case of Theorem 1 of~\cite{decoupling}]\label{thm:decoupling}
	Let $X_1,\dots, X_n$, $Y_1,\dots,Y_n$ are independent random variables on a measurable space over $S$, where $X_i$ and $Y_i$ has the same distribution for $i = 1,\dots,n$. Let $f_{ij}(\cdot,\cdot)$ be a family of functions taking $S\times S$ to a Banach space $(B,\|\cdot\|)$. Then there exists absolute constant $C$, such that for all $n\ge 2$, $t>0$, 
	$$\Pr\left[\left\|\sum_{i\neq j}f_{ij}(X_i,X_j)\right\|\ge t\right] \le C\Pr\left[\left\|\sum_{i \neq j}f_{ij}(X_i,Y_j)\right\|\ge t/C\right]$$
\end{theorem}

We also need the Matrix Bernstein's Inequality:

\begin{theorem}[Matrix Bernstein, \cite{tropp2012user}] \label{thm:matbernstein}Consider a finite sequence $\{X_k\}$ of independent, random symmetric matrices with dimension $d$. Assume that each random matrix satisfies
$$
\E[X_k] = 0 \mbox{ and } \|X_k\| \le R \mbox{ almost surely.}
$$
Then, for all $t\ge 0$,
$$
\Pr[\|\sum_k X_k\| \ge t] \le d\cdot \exp\left(\frac{-t^2/2}{\sigma^2+Rt/3}\right) \mbox{ where }\sigma^2 := \|\sum_k \E[X_k^2]\|.
$$
\end{theorem}

\section{Sum-of-Square Proofs}

In this section we state some lemmas that can be proved by low-degree SoS proofs. Most of these lemmas can be found in~\cite{BarakS14} and~\cite{BarakKS14} but we still give the proofs here for completeness.

%\Tnote{We need a lemma about $x^TMx\sosle \|M\|\|x\|^2$ and a lemma about cauchy-shwarz}

\begin{lemma}\label{lem:sos-cauchy-schwarz}[SoS proof for Cauchy-Schwarz]
	Cauchy-Schwarz inequality can be proved by degree-2 sum of squares proofs, 
	$$\left(\sum_{i=1}^{n} a_i^2\right) \left(\sum_{i=1}^{n} b_i^2\right) - \left(\sum_i a_ib_i\right)^2 = \sum_{i,j}(a_ib_j-a_jb_i)^2$$
\end{lemma}

\begin{lemma}\label{lem:sos-holder}
	For any vector $x$, $y$, we have that for even number $k$, 
	$$\|x+y\|^{k} \sosle_k 2^{k-1}(\|x\|^k + \|y\|^k)$$
\end{lemma}

\begin{proof}
	Note that it suffices to prove it for one dimensional vector $x,y$. We prove by induction. For $k = 2$, it just follows Cauchy-Schwarz. Suppose it is true for $k-2$ case, we have 
	$$(x+y)^{k} = (x+y)^{k-2}(x+y)^2 \sosle 2^{k-3}(x^{k-2}+y^{k-2}) \cdot 2(x^2+y^2) $$
	Note that $$2(x^k+y^k) -(x^{k-2}+y^{k-2})(x^2+y^2) = (x^2-y^2)^2(x^{k-4}+x^{k-6}y^2 +\dots+y^{k-4})\sosge 0$$
	
	Combing the two equations above we obtain the desired result. 
\end{proof}

\begin{lemma}\label{lem:sosspectral}
	Suppose $M$ is $m\times n$ matrix with spectral norm $\|M\|$, then 
	$$(x^TMy)^2 \sosle_{4} \|x\|^2\|y\|^2 \|M\|^2$$
\end{lemma}

\begin{proof}
	Assume $m\le n$ without loss of generality, and suppose $M$ has singular decomposition $M = U\Sigma V^T$ where $\Sigma = \diag(\sigma_1,\dots,\sigma_m)$. Let $z = x^TU$ and $w = V^Ty$. Then 
	
	$$(x^TMy)^2  = \left(\sum_{i=1}^m \sigma_iz_iw_i\right)^2 \sosle_4\left(\sum_{i=1}^m \sigma_i^2z_i^2\right) \left(\sum_{i=1}^m w_i^2\right) \le \|M\|^2 \|z\|^2 \|w\|^2 = \|x\|^2\|y\|^2 \|M\|^2$$
\end{proof}

\begin{lemma}\label{lem:square-root}
For a nonnegative real number $a$ and a set of polynomial $R$ and positive integer $k$, if a polynomial $p(x)$ satisfy $p(x)\sosle_{R,k} a^2$, then $p(x) \sosle_{R,k'} a$ for $k' = \max\{k,2\deg(p)\}$. 
\end{lemma}

\begin{proof}
By a simple manipulation of algebra, we have that 
$$p(x) - a \sosle_{R,k}\frac{1}{2a}(p(x)-a)^2\sosle_{R,k'}0.$$
\end{proof}

%\bibliographystyle{alpha}
%\bibliography{ref}
\end{document}